\documentclass[a4paper, conference, final]{IEEEtran}

\usepackage{cite}   
\usepackage{color}   
\usepackage[caption=false, font=footnotesize]{subfig}
\usepackage[pdftex]{graphicx}
\usepackage[latin1]{inputenc}
\usepackage{url}       
\usepackage{amsmath,amsfonts,amssymb,amsthm}
\usepackage{times}
\usepackage{algorithm}
\usepackage{algpseudocode}

\usepackage{enumerate}
\usepackage{siunitx}
\usepackage{multirow}
\usepackage{tikz}

\newtheorem{lemma}{Lemma}
\newtheorem{proposition}{Proposition}
\newtheorem{corollary}{Corollary}

\theoremstyle{remark}
\newtheorem*{remark}{Remark}

\algnewcommand\True{\textrm{TRUE}}
\algnewcommand\False{\textrm{FALSE}}
\algnewcommand\LeftmostOne{\textrm{LeftmostOne}}
\algnewcommand\Degree{\textrm{Degree}}
\algnewcommand\Diag{\textrm{Diag}}
\algnewcommand\Swap{\textrm{Swap}}

\newcommand{\SLIF}[2]{\State \algorithmicif\ {#1}\ \algorithmicthen\ {#2}}

\newcommand\Algphase[1]{%
\vspace*{-0.7\baselineskip}\Statex\hspace*{\dimexpr-\algorithmicindent-2pt\relax}\rule{\columnwidth}{0.4pt}%
\vspace*{-0.8mm}\Statex\hspace*{-\algorithmicindent}{#1}%
\vspace*{-0.3\baselineskip}\Statex\hspace*{\dimexpr-\algorithmicindent-2pt\relax}\rule{\columnwidth}{0.4pt}%
}

\begin{document}
\title{Binary Systematic Network Coding\\for Progressive Packet Decoding}

\author{\IEEEauthorblockN{Andrew L. Jones, Ioannis Chatzigeorgiou and Andrea Tassi}
\IEEEauthorblockA{School of Computing and Communications, Lancaster University, United Kingdom\\
Email: \{a.jones2, i.chatzigeorgiou, a.tassi\}@lancaster.ac.uk}}

\maketitle

\begin{abstract}
We consider binary systematic network codes and investigate their capability of decoding a source message either in full or in part. We carry out a probability analysis, derive \mbox{closed-form} expressions for the decoding probability and show that systematic network coding outperforms conventional network coding. We also develop an algorithm based on Gaussian elimination that allows progressive decoding of source packets. Simulation results show that the proposed decoding algorithm can achieve the theoretical optimal performance. Furthermore, we demonstrate that systematic network codes equipped with the proposed algorithm are good candidates for progressive packet recovery owing to their overall decoding delay characteristics.
\end{abstract}

\vspace{1mm}

% -------------------------------------

\section{Introduction}
\label{sec:intro}

Network coding (NC), originally proposed in \cite{Ahlswede00}, has the potential to significantly improve network reliability by mixing packets at a source node or at intermediate network nodes prior to transmission. The classical implementation of NC, which is often referred to as \textit{straightforward NC} \cite{Zhang06}, randomly combines source packets using finite field arithmetic. As the size of the field increases, the likelihood of the transmitted packets being linearly independent also increases. However, the decoding process at the receiver is computationally expensive, especially if the field size is large. Furthermore, straightforward NC incurs a substantial decoding delay because source packets can be recovered at the receiver only if the received network-coded packets are at least as many as the source packets.

Heide \textit{et al.}~\cite{Heide09} proposed the adoption of binary systematic NC, which operates over a finite field of only two elements, as a means of reducing the decoding complexity of straightforward NC. A source node using systematic NC first transmits the original source packets and then broadcasts linear combinations of the source packets. The reduction in decoding complexity at the receiver decreases energy consumption and makes systematic NC suitable for energy-constrained devices, such as mobile phones and laptops. Lucani \textit{et al.}~\cite{Lucani10} developed a Markov chain model to show that the decoding process of systematic NC in time division duplexing channels requires considerably fewer operations, on average, than that of straightforward NC. Barros \textit{et al.} \cite{Barros09} and Prior and Rodrigues~\cite{Prior11} observed that opting for systematic NC as opposed to straightforward NC reduces decoding delay without sacrificing throughput. Therefore, systematic network codes exhibit desirable characteristics for multimedia broadcasting and streaming applications. More recently, Saxena and V\'{a}zquez-Castro \cite{Saxena13} discussed the advantages of systematic NC for transmission over satellite links.

As in~\cite{Heide09}, we also consider binary systematic network codes and investigate their potential in delivering services, such as multimedia and streaming, which often require the progressive recovery of source packets and the gradual refinement of the source message. Our objective is to prove that systematic NC not only exhibits a lower decoding complexity than straightforward NC, as shown in~\cite{Lucani10}, but also a better performance, as observed in~\cite{Barros09}. Even though our focus is on binary systematic NC, we explain that our analysis can be easily extended to finite fields of larger size. In addition, we develop a decoding algorithm and propose a framework, which helps us study the performance of systematic NC in terms of the probability of recovering a source message either in part or in full.  

The rest of the paper has been organised as follows. Section~\ref{sec:sys_net_coding} analyses the performance of systematic NC and introduces metrics for evaluating its capability of progressively recovering source messages. Section~\ref{sec:prog_decoding} proposes a modification to the Gaussian elimination algorithm that allows source packets to be progressively decoded. Section~\ref{sec:results} discusses the computational cost and accuracy of the proposed decoding algorithm, validates the derived theoretical expressions and contrasts the performance of systematic NC with that of benchmark transmission schemes. The main contributions of the paper are summarised in Section~\ref{sec:conclusions}.

% -------------------------------------

\section{Binary Systematic Network Coding}
\label{sec:sys_net_coding}

Let us consider a source node, which segments a message $\mathbf{s}=[\,s_i\,]_{i=1}^{K}$ into $K$ source packets and encodes them using a systematic NC encoder. The encoder generates and transmits $N$ packets, which comprise $K$ systematic packets followed by \mbox{$N-K$} coded packets. The systematic packets are identical to the source packets, while the coded packets are obtained by linearly combining source packets. The $n$-th transmitted packet, denoted by $t_n$, can be expressed as follows
\begin{equation}
t_n = \begin{cases}  
\:s_n &\text{if $n \leq K$} \vspace{1mm}\\
\:\displaystyle\sum_{i=1}^{K}g_{n,i}\:s_{i} &\text{if $K<n\leq N$} 
\end{cases}
\label{eq.encoded_SysNC}
\end{equation}
where $g_{n,i}$ is a binary coefficient chosen uniformly at random from the elements of the finite field $\textrm{GF}(2)$. We can also express $t_n$ in matrix notation as $t_n=\mathbf{G}[\,n\,]\cdot\mathbf{s}^\intercal$, where $\mathbf{G}[\,n\,]=[\,g_{n,i}\,]_{i=1}^K$ is the coding vector associated with $t_n$. Note that when $n\leq K$, in line with the definition of binary systematic NC, we set $g_{n,i}=1$ if $i=n$, else $g_{n,i}=0$.

In the remainder of this section, we investigate the theoretical performance of systematic NC and derive analytical expressions for the probability of decoding the entire source message or a fraction of the source message. We also present performance metrics and benchmarks for the evaluation of systematic NC for progressive packet recovery.

\subsection{Probability of Decoding the Entire Source Message}
\label{subsec:FullProb}

As previously mentioned, a source node using systematic NC transmits $N$ packets, of which $K$ are systematic and the remaining \mbox{$N-K$} are coded. Assume that a receiver successfully recovers \mbox{$r$} packets, of which \mbox{$h$} are systematic and \mbox{$r-h$} are coded. The coding vectors of the $r$ received packets are stacked to form the $r\times K$ \textit{decoding matrix} $\mathbf{G}$. 

%Without loss of generality, the columns of $\mathbf{G}$ can be rearranged and the reorganized matrix can be partitioned into submatrices as follows 
%\begin{equation}
%\mathbf{G} =
%\left[\arraycolsep=5pt\def\arraystretch{1.4}
%\begin{array}{c|c}
%\mathbf{G}_{r,h} &\mathbf{G}_{r,K-h}
%\end{array}
%\right] = 
%\left[\arraycolsep=5pt\def\arraystretch{1.4}
%\begin{array}{c|c}
% \mathbf{1}_{h,h} &\mathbf{0}_{h, K-h} \\
% %\hline
% \mathbf{G}_{r-h,h} &\mathbf{G}_{r-h,K-h}
%\end{array}
%\right]
%\label{eq.P}
%\end{equation}
%where $\mathbf{1}_{h,h}$ is the \mbox{$h\times h$} identity matrix and $\mathbf{0}_{h, K-h}$ is a \mbox{$h \times (K-h)$} zero matrix.

Let $f_{K}(r,N)$ denote the probability of decoding the $K$ source packets given that $r$ packets have been received. We understand that $f_{K}(r,N)$ is non-zero only if $K\leq r\leq N$. The value of $r$ also determines the smallest allowable value of $h$. For instance, if \mbox{$K\leq N<2K$}, the $N-K$ transmitted coded packets are fewer than the $K$ transmitted systematic packets; given that $r\geq K$ packets are received, the number of received systematic packets $h$ should be at least \mbox{$r\!-\!(N\!-\!K)$}. Otherwise, if \mbox{$N\geq 2K$}, the smallest value of $h$ can be zero. Therefore, $h$ is defined in the range \mbox{$\max{(0,\,r\!-\!N\!+\!K)}\leq h\leq K$}. Having defined the parameters of the system model and their interdependencies, we can now proceed with the derivation of an analytical expression for $f_{K}(r,N)$.

\begin{lemma}
\label{lemma.cond_prob_q2}
For \mbox{$N\geq K$} transmitted packets, the probability of a receiver decoding all of the $K$ source packets, given that \mbox{$K\leq r\leq N$} packets have been successfully received, is
\begin{equation}
\label{eq.cond_prob_q2}
f_{K}(r,N)\!=\!\frac{\binom{N-K}{r-K}\!+\!\!\!\!{\displaystyle\sum_{h = h_{\min}}^{K-1}}\!\!\!\binom{K}{h}\!\binom{N-K}{r-h}\!\!\!{\displaystyle\prod_{j = 0}^{K-h-1}}\!\!\!\!\left(1-2^{-r+h+j}\right)}{\textstyle\binom{N}{r}}\!
\end{equation}where $h_{\min}=\max{(0,\,r\!-\!N\!+\!K)}$.
%\setlength{\arraycolsep}{0.0em}
%\begin{eqnarray}
%\label{eq.sum_two_terms}
%f_{K}(r,N)&{}={}& \displaystyle\binom{N}{r}^{-1} \Bigg[ \displaystyle\binom{N-K}{r-K}+\nonumber\\
%&&\hspace{-9mm}+{\displaystyle\sum_{h = h_{\min}}^{K-1}}\displaystyle\binom{K}{h}\displaystyle\binom{N-K}{r-h}{\displaystyle\prod_{j = 0}^{K-h-1}}\left(1-2^{-r+h+j}\right)\Bigg]
%\end{eqnarray}where $h_{\min}=\max{(0,\,r\!-\!N\!+\!K)}$.
%\setlength{\arraycolsep}{5pt}
\end{lemma}

\begin{proof}
The decoding probability $f_{K}(r,N)$ can be decomposed into the sum of the following probabilities
\begin{equation}
\label{eq.sum_two_terms}
f_{K}(r,N)=\mathbb{P}\{h\!=\!K\}\,+\!\!\sum_{h=h_{\min}}^{K-1}\!\!\!\mathbb{P}\{h\!<\!K\}\; w_{K-h}(r-h).
\end{equation}
The term $\mathbb{P}\{h=K\}$ represents the probability of recovering the $K$ source packets directly from the $K$ successfully received systematic packets. This is the case when \mbox{$r-K$} out of the \mbox{$N-K$} coded packets have been successfully delivered to the receiver along with the \mbox{$K$} systematic packets. Considering that $r$ out of the $N$ transmitted packets have been received, we can deduce that $\mathbb{P}\{h=K\}$ is given by
\begin{equation}
\label{eq.term1}
\mathbb{P}\{h\!=\!K\} = \frac{\binom{N-K}{r-K}}{\binom{N}{r}}. 
\end{equation}
The sum of products in \eqref{eq.sum_two_terms} considers the probability of recovering \mbox{$h<K$} systematic packets and decoding the remaining \mbox{$K-h$} source packets from the \mbox{$r-h$} received coded packets. More specifically, the probability $\mathbb{P}\{h\!<\!K\}$ of receiving $h$ out of the $K$ systematic packets and \mbox{$r-h$} out of the \mbox{$N-K$} coded packets is equal to
\begin{equation}
\label{eq.term2.1}
\mathbb{P}\{h<K\} = \frac{\binom{K}{h}\binom{N-K}{r-h}}{\binom{N}{r}}. 
\end{equation}
On the other hand, the probability of having \mbox{$K-h$} linearly independent coded packets among the \mbox{$r-h$} received ones can be obtained from the literature of straightforward NC, for example \cite{Trullols}. We find that
\begin{equation}
\label{eq.term2.2}
w_{K-h}(r-h)=\prod_{j = 0}^{K-h-1}\left(1-2^{-(r-h)+j}\right).
\end{equation}
Substituting \eqref{eq.term1}, \eqref{eq.term2.1} and \eqref{eq.term2.2} into \eqref{eq.sum_two_terms} gives \eqref{eq.cond_prob_q2}. This concludes the proof.
\end{proof}

\begin{proposition}
\label{prop.DecAllK}
The probability of a receiver decoding all of the $K$ source packets, after the transmission of $N\geq K$ packets over a channel characterized by a packet erasure probability~$p$, can be expressed as follows
\begin{equation}
\label{eq.P_full}
\mathrm{P}_{K}(N)=\sum_{r=K}^{N}{\textstyle \binom{N}{r}}\left(1-p\right)^{r}p^{N-r}f_{K}(r,N).
\end{equation}
\end{proposition}

\begin{proof}
The proof follows from Lemma \ref{lemma.cond_prob_q2}. The conditional probability $f_{K}(r,N)$ has been weighted by the probability of successfully receiving $r$ out of $N$ transmitted packets and averaged over all valid values of $r$.
\end{proof}

The closed-form expressions for the decoding probability of systematic network codes can be used to contrast their performance to the performance of straightforward network codes and give rise to the following proposition.

\begin{proposition}
\label{prop.sysNC}
Systematic network codes exhibit a higher probability of decoding all of the $K$ packets of a source message than straightforward network codes.
\end{proposition}

\begin{proof}
For the same number of received packets $r$, the probability of decoding all of the $K$ source packets is $f_{K}(r,N)$ for systematic NC and $w_{K}(r)$ for straightforward NC, where \mbox{$w_{K}(r)=\prod_{j = 0}^{K-1}\left(1-2^{-r+j}\right)$} as per \eqref{eq.term2.2}. If we show that the relationship \mbox{$f_{K}(r,N)\geq w_{K}(r)$} holds for all valid values of $N$, we can infer that the decoding probability of systematic NC is higher than that of straightforward NC. Dividing $f_{K}(r,N)$ by $w_{K}(r)$ gives
\begin{equation}
\label{eq.DecProbRatio}
\frac{f_{K}(r,N)}{w_{K}(r)}={\textstyle \binom{N}{r}^{\!-1}}\!\biggl[{\textstyle \binom{N-K}{r-K}}A\,+\!\!\sum_{h=h_{\min}}^{K-1}\!\!\!{\textstyle \binom{K}{h}\!\binom{N-K}{r-h}}B_h\biggr]
\end{equation}
where
\begin{equation}
\label{eq.DefAandB}
A \! = \!\!\prod_{j=0}^{K-1}\!\frac{2^{r-j}}{2^{r-j}-1}
\,\,\textrm{and}\,\,
B_h \! = \!\begin{cases}
1, &\!\!\!\textrm{for }h=0\vspace{1mm}\\
{\displaystyle \prod_{j=0}^{h-1}\!\frac{2^{r-j}}{2^{r-j}-1}}, &\!\!\!\textrm{for }h>0.
\end{cases}
\end{equation}
Note that $A>1$ and $B_h\geq1$ for all valid values of $r$, that is, $K\leq r\leq N$. Therefore, the right-hand side of \eqref{eq.DecProbRatio} can become a lower bound on the ratio $f_{K}(r,N)/w_{K}(r)$ if coefficients $A$ and $B_{h}$ are removed. More specifically, we can obtain
\begin{equation}
\label{eq.BoundOnDecProbRatio}
\frac{f_{K}(r,N)}{w_{K}(r)}>{\textstyle \binom{N}{r}}^{\!-1}
\sum_{h=h_{\min}}^{K}\!\!{\textstyle \binom{K}{h}}{\textstyle\binom{N-K}{r-h}}
\end{equation}if the binomial coefficient $\binom{N-K}{r-K}$ in \eqref{eq.DecProbRatio} is included into the sum and the upper limit of the sum is updated accordingly. We distinguish the following two cases for the value of $N$:
\begin{itemize}
\item $N\geq 2K$: In this case, we have $h_{\min}=0$. Invoking a special instance of the Chu-Vandermonde identity \mbox{\cite[\!p.\!\,\,41]{Koepf}}, we can reduce the sum at the right-hand side of \eqref{eq.BoundOnDecProbRatio} to
\begin{equation}
\label{eq.CaseNgr2K}
\sum_{h=0}^{K}{\textstyle \binom{K}{h}}{\textstyle\binom{N-K}{r-h}}={\textstyle \binom{N}{r}}.
\end{equation}
\item $K\!\leq\!N\!<\!2K$: As previously explained, \mbox{$h_{\min}\!=\!r\!-\!N\!+\!K$}. Setting \mbox{$h'\!=\!N\!-\!K\!-\!r\!+\!h$}, expressing the sum in \eqref{eq.BoundOnDecProbRatio} in terms of $h'$, exploiting the properties of binomial coefficients and using the widely-known Vandermonde's convolution \cite[\!p.\!\,\,29]{Roman} gives
\begin{equation}
\label{eq.CaseNls2K}
\begin{split}
\sum_{h=r-N+K}^{K}\!\!{\textstyle \binom{K}{h}}{\textstyle\binom{N-K}{r-h}}&=
\sum_{h'=0}^{N-r}{\textstyle \binom{K}{N-r-h'}}{\textstyle\binom{N-K}{h'}}\\
&={\textstyle \binom{N}{N-r\,}}={\textstyle \binom{N}{r}}.
\end{split}
\end{equation}
\end{itemize}
If we combine identities \eqref{eq.CaseNgr2K} and \eqref{eq.CaseNls2K} with inequality \eqref{eq.BoundOnDecProbRatio}, we obtain \mbox{$f_{K}(r,N)/w_{K}(r)>1$} for all valid values of $N$, which concludes the proof. We note that the ratio \mbox{$f_{K}(r,N)/w_{K}(r)$} approaches $1$ as the value of $N-K$ increases.
\end{proof}

\begin{remark}
Even though this paper is concerned with binary systematic NC, i.e. the elements of matrix $\mathbf{G}$ are selected uniformly at random from $\mathrm{GF}(2)$, the same reasoning can be employed to obtain $\mathrm{P}_{K}(N)$ when operations are performed over $\mathrm{GF}(q)$ for $q\geq 2$. The probability $f_{K}(r,N)$ of decoding the entire source message, given that $r$ packets have been received, can be written as
\begin{equation}
\label{eq.cond_prob_any_q}
f_{K}(r,N)\!=\!\frac{\binom{N-K}{r-K}\!+\!\!\!\!\!{\displaystyle\sum_{h = h_{\min}}^{K-1}}\!\!\!\!\binom{K}{h}\!\binom{N-K}{r-h}\!\!\!{\displaystyle\prod_{j = 0}^{K-h-1}}\!\!\!\!\!\left(1\!-\!q^{-r+h+j}\right)}{\textstyle\binom{N}{r}}\!.
\end{equation}
Both Propositions \ref{prop.DecAllK} and \ref{prop.sysNC} hold for $q\geq2$. Substituting \eqref{eq.cond_prob_any_q} into \eqref{eq.P_full} gives the general expression for $\mathrm{P}_{K}(N)$.
\end{remark}

\subsection{Probability of Decoding a Fraction of the Source Message}

In Section \ref{subsec:FullProb}, we focused on deriving the probability of decoding the $K$ source packets when \mbox{$N\geq K$} packets have been transmitted. Of equal interest is the probability of recovering at least \mbox{$M<K$} source packets when $N\geq M$ packets have been transmitted. To the best of our knowledge, a closed-form expression for this probability, denoted hereafter as $\mathrm{P}_{K,M}(N)$, has not been obtained for straightforward NC. However, a good approximation, which follows readily from Proposition \ref{prop.DecAllK}, can be computed for the case of systematic NC.

%A closed-form expression for this probability, denoted hereafter by $\mathrm{P}_{K,M}(N)$, has not been obtained for straightforward NC to the best of our knowledge.

\begin{corollary}
The probability of recovering at least \mbox{$M<K$} source packets, when \mbox{$N\geq M$} packets have been transmitted over a channel with packet erasure probability $p$, can be approximated by
\begin{equation}
\label{eq.P_partial}
\mathrm{P}_{K,M}(N)\approx\sum_{r=M}^{N_{\min}}{\textstyle \binom{N_{\min}}{r}}\left(1-p\right)^{r}p^{N_{\min}-r}
\end{equation}where $N_{\min}=\min{(K,N)}$.
\end{corollary}

\begin{proof}
The number of transmitted systematic packets is either $N$ if $N<K$, or $K$ if $N\geq K$. In general, $\min{(K,N)}$ systematic packets are sent over the packet erasure channel, for any value of $N$. If we wish to recover at least $M<\min{(K,N)}$ source packets and the erasure probability $p$ is small, $M$ or more received packets will most likely be systematic and, thus, linearly independent. As a result, the probability of decoding at least $M$ source packets reduces to the probability of recovering at least $M$ systematic packets, given by \eqref{eq.P_partial}.
\end{proof}

We remark that the assumption of a low value of $p$ is reasonable when the physical layer employs error correcting codes that improve the channel conditions as ``seen'' by higher network layers, where NC is usually applied. For example, the Long Term Evolution Advanced (LTE-A) framework considers an erasure probability of $p=0.1$ \cite{sesia}.

\subsection{Performance Metrics and Benchmarks}

In order to assess the performance of systematic NC and explore its capability to progressively decode a source message, we will compare it with \textit{ordered uncoded} (OU) transmission~\cite{Jones} and straightforward NC. In OU transmission, the $K$ source packets are periodically repeated. The transmitted packet at time step $n=i+m K$ can be expressed as \mbox{$t_{i+m K} = s_i$} for \mbox{$i=1,\ldots,K$} and $m\geq0$. We note that transmission is \textit{uncoded} in the sense that transmitted packets are not linear combinations of the source packets. By contrast, the $n$-th transmitted packet in straightforward NC is given by \mbox{$t_n=\sum_{i=1}^{K}g_{n,i}\:s_{i}$} for $n>0$, implying that all transmitted packets are linear combinations of the source packets.

Probabilities $\mathrm{P}_{K,M}(N)$ and $\mathrm{P}_{K}(N)$ will be used to contrast the performance of systematic NC, straightforward NC and OU transmission. In order to create links between the two decoding probabilities, we introduce the following parameters:
\begin{itemize}
\item $\hat{P}$ is a predetermined target probability of packet recovery that a transmission scheme has to attain. Probabilities $\mathrm{P}_{K,M}(N_1)$ and $\mathrm{P}_{K}(N_2)$ can be set equal to $\hat{P}$ in order to determine the number of transmitted packets $N_1$ and $N_2$ that are required for the partial or full recovery of the source message, respectively.
\item $\hat{N}$ signifies the minimum number of transmitted packets required by the receiver to recover at least $M$ source packets with a probability of at least $\hat{P}$.
\item $\Delta N$ denotes the minimum number of additional packets that should be transmitted so that the receiver recovers the $K$ source packets with a probability of at least $\hat{P}$.
\end{itemize}

A performance comparison of the investigated schemes will be carried out in Section \ref{sec:results}. Prior to that, we discuss decoding algorithms for NC schemes and propose a decoding process that allows progressive decoding of source packets in the following section.

% -------------------------------------

\section{Progressive Decoding}
\label{sec:prog_decoding}

If the objective of the decoding algorithm is the recovery of the $K$ source packets after the reception of at least $K$ transmitted packets, Gaussian Elimination (GE) could be used especially when the value of $K$ is small. The GE algorithm transforms the decoding matrix $\mathbf{G}$ into row-echelon form \cite{Epperson}. The rank of the transformed matrix, which is equal to the rank of the original decoding matrix, can be obtained by inspecting the number of non-zero rows within the echelon form. If the rank is $K$, that is, if $\mathbf{G}$ is a \textit{full-rank} matrix, the $K$ source packets can be successfully recovered.

GE and schemes based on Belief Propagation (BP) \cite{JBPGE} experience a large spike in computation when $K$ transmitted packets are received. On-the-Fly Gaussian Elimination (OFGE) \cite{Bioglio} manages to mitigate the decoding delay and computational complexity of GE by invoking an optimized triangulation process \textit{every time} a packet is received. The OFGE decoder spreads computation out over each packet arrival and the decoding matrix $\mathbf{G}$ is already in partial triangular form by the time the $K$-th transmitted packet is received.

Both GE and OFGE have been designed to perform \mbox{full-rank} decoding. As a result, if the rank of $\mathbf{G}$ is less than $K$, that is, if the decoding matrix is \mbox{\textit{rank-deficient}}, some source packets might still be decodable but GE or OFGE will not necessarily identify them. A modified version of OFGE, which we refer to as OFGE for Progressive Decoding \mbox{(OFGE-PD)}, was presented in \cite{Jones}. Similarly to OFGE, \mbox{OFGE-PD} also comprises a triangulation stage and a back-substitution stage. An additional stage, called the XORing phase, enables \mbox{OFGE-PD} to decode source packets from \mbox{rank-deficient} decoding matrices at the expense of increased computational complexity.

We revisited the original GE algorithm and we amalgamated the OFGE principle of initiating the decoding process whenever a packet is received. A sketch of the proposed algorithm, referred to as Gaussian Elimination for Progressive Decoding \mbox{(GE-PD)}, is presented in Algorithm \ref{GE-PD}. To facilitate the description of \mbox{GE-PD}, we introduced function $\texttt{Degree}$, which determines the number of non-zero elements in a row vector; function $\texttt{Diag}$, which generates a row vector containing the elements of the main diagonal of a matrix; function $\texttt{LeftmostOne}$, which returns the position of the first non-zero entry in a row vector; and function $\texttt{Swap}$, which swaps two rows in a matrix. The decoding matrix $\mathbf{G}$ is initially set equal to the $K\times K$ zero matrix. Recall that $\mathbf{G}[\,n\,]$ represents the $n$-th row of $\mathbf{G}$, while $\mathbf{G}[\,i\,][\,j\,]$ denotes the entry of $\mathbf{G}$ in the $i$-th row and $j$-th column (equivalent to $g_{i,j}$). We note that, depending on the adopted programming language, the code can be further optimized and the execution speed of \mbox{GE-PD} improved.

As line~\ref{alg.R} in Algorithm \ref{GE-PD} indicates, whenever a new coding vector $\mathbf{R}$ is received, it is updated so that any previously decoded source packets are not considered again in the decoding process. If the updated row-vector $\mathbf{R}$ still contains non-zero entries, it is appended to the bottom of the decoding matrix $\mathbf{G}$ (lines~\ref{alg.AppendRowStart}-\ref{alg.AppendRowStop}). Lines~\ref{alg.ArrangeStart}-\ref{alg.ArrangeStop} rearrange the rows of $\mathbf{G}$ in an effort to transform it into an upper triangular matrix. Lines~\ref{alg.TransStart}-\ref{alg.TransStop} aim to transform $\mathbf{G}$ into row-echelon form by ensuring that each non-zero element on the main diagonal of $\mathbf{G}$ is the only non-zero element in that column. Finally, function $\texttt{BackSubstitution}$ is called in line~\ref{alg.backsub} to establish which source packets are decodable. The efficiency and accuracy of \mbox{GE-PD} are investigated in the following section.
%The final section (lines 26-33) ~\ref{alg.backsub} establishes which source packets are decodable. The efficiency and accuracy of \mbox{GE-PD} are investigated in the following section.

\begin{algorithm}[t]
\caption{Gaussian Elimination for Progressive Decoding}
\label{GE-PD}
\begin{algorithmic}[1]
\scriptsize
\State Receive new $1\times K$ coding vector $\mathbf{R}$
\State Set entries in $\mathbf{R}$ that correspond to decoded packets to 0\label{alg.R}
%\For{$i=1$ to $K$}
%	\SLIF{$\mathbf{D}[\,i\,]=1$}{$\mathbf{R}[\,i\,] \gets 0$}
%\EndFor
\If {$(\Degree(\mathbf{R}) > 0)$}\label{alg.AppendRowStart}% \And (\Degree(\mathbf{D}) < K)$}
    \State $\mathbf{G}[\,K+1\,] \gets \mathbf{R}$\label{alg.AppendRowStop}
    \For{$i = 1$ to $K$}
		\State $one\_in\_diag \gets \True$\label{alg.ArrangeStart}
		\If {$(\mathbf{G}[\,i\,][\,i\,] = 0)$}
			\State $one\_in\_diag \gets \False,\;\;j \gets i+1$
			\Repeat
				\If {$\LeftmostOne(\mathbf{G}[\,j\,])=i$}
					\State $\Swap(\mathbf{G}[\,i\,], \mathbf{G}[\,j\,])$				
					\State $one\_in\_diag \gets \True$					
             	\EndIf
				\State $j \gets j+1$			
			\Until{$(\,j>K+1\,)\;\;\textbf{or}\;\;one\_in\_diag$}                   		
%            \For {$j = (i+1)$ to $(K+1)$}
%             	\If {$\LeftmostOne(\mathbf{G}[\,j\,])=i$}
%					\State $\Swap(\mathbf{G}[\,i\,], \mathbf{G}[\,j\,])$
%					\State $one\_in\_diag \gets \True$					
%					\State $\mathbf{break}$
%             	\EndIf
%			\EndFor
		\EndIf\label{alg.ArrangeStop}
		\If{$one\_in\_diag$}\label{alg.TransStart}
			\For {$j = 1$ to $(K+1)$}
				\If{$(\,j \neq i\,)\;\;\textbf{and}\;\;(\mathbf{G}[\,j\,][\,i\,]=1)$}
					\State $\mathbf{G}[\,j\,] \gets \mathbf{G}[\,j\,]\oplus\mathbf{G}[\,i\,]$
				\EndIf
			\EndFor
		\EndIf\label{alg.TransStop}
	\EndFor   
%%    \If {$\Degree(\Diag(\mathbf{G})) <= K$}
%		\For {$i = K$ to $1$ step $-1$}
%			\If {$\Degree(\mathbf{G}[\,i\,]) = 1$}
%				\State $j = \LeftmostOne(\mathbf{G}[\,i\,])$
%				\For {$k = 1$ to $K$}
%					\SLIF {$k \neq i$}{$\mathbf{G}[\,k\,][\,j\,] \gets 0$}
%				\EndFor
%				\State $\mathbf{D}[\,j\,] \gets 1$
%			\EndIf
%		\EndFor
%%    \EndIf
	\State $\mathbf{G} \gets $ BackSubstitution($\mathbf{G}$, $K$)\label{alg.backsub}
    \State $\mathbf{G} \gets $ Top $K$ rows of $\mathbf{G}$
\EndIf
\end{algorithmic}
\begin{algorithmic}[1]
\Algphase{\textbf{Function} BackSubstitution($\mathbf{G}$, $K$)}
\scriptsize
	\For {$i = K$ to $1$ step $-1$}
		\If {$\Degree(\mathbf{G}[\,i\,]) = 1$}
			\State $j = \LeftmostOne(\mathbf{G}[\,i\,])$
			\For {$k = 1$ to $K$}
				\SLIF {$k \neq i$}{$\mathbf{G}[\,k\,][\,j\,] \gets 0$}
			\EndFor
		\EndIf
	\EndFor
    \State \Return $\mathbf{G}$
\end{algorithmic}
\end{algorithm}
%\vspace{-1.5mm}
% -------------------------------------

\section{Results and Discussion}
\label{sec:results}

This section compares the proposed \mbox{GE-PD} with \mbox{OFGE-PD}, OFGE and GE in terms of computational cost and capability of progressively recovering source packets. The decoding algorithm that achieves the best accuracy but requires the least computational time is identified. It is then used to obtain simulation results, which are compared to theoretical predictions in order to validate the derived analytical expressions for systematic NC. The performance of systematic NC is then contrasted to that of straightforward NC and OU transmission, and the suitability of each scheme for progressive packet recovery is discussed.

\subsection{Assessment of GE-PD}

Fig.~\ref{fig.time_comp} compares the \textit{computational cost} of the considered decoding schemes. Recall that \mbox{GE-PD} and \mbox{OFGE-PD} are modified versions of GE and OFGE, respectively, which have been adapted to recover source packets from rank-deficient decoding matrices, as described in Section \ref{sec:prog_decoding}. The computational cost has been expressed in terms of the time required for a decoder to recover the full sequence of $K$ source packets when straightforward NC is applied and channel conditions are perfect, i.e. $p=0$. The plotted results were obtained on a simulation platform equipped with an Intel Core i7-3770 processor and 8 GB of RAM. As expected \cite{Bioglio}, Fig.~\ref{fig.time_comp} shows that OFGE yields substantial computational savings over the conventional GE. However, the inclusion of progressive decoding capabilities in OFGE adds noticeable overhead to the decoding process. We observe that the computational cost of the resultant \mbox{OFGE-PD} increases rapidly for large values of $K$. On the other hand, GE-PD is not only more efficient than the original GE but also executes faster than OFGE.

Straightforward NC for $K=20$ source packets and perfect channel conditions were also assumed for the performance assessment of the four decoding schemes. Fig.~\ref{fig.dec_prob} depicts the probability of each scheme recovering at least half \mbox{($M=10$)} or all \mbox{($M=20$)} of the source packets when $N$ packets have been transmitted. As we see, OFGE is not optimized for recovering a fraction of the source message in contrast to \mbox{OFGE-PD}, which requires a smaller number of transmitted packets to recover half of the source message but at a higher computational cost. A fact worthy of attention is that the decoding accuracy of GE is matched by that of \mbox{GE-PD}, which exhibits a computational cost as low as that of OFGE. For this reason, the proposed \mbox{GE-PD} was the decoding algorithm of choice for the simulation of the considered NC-based schemes. 

\begin{figure}[t]
\centering
\includegraphics[width=0.95\columnwidth]{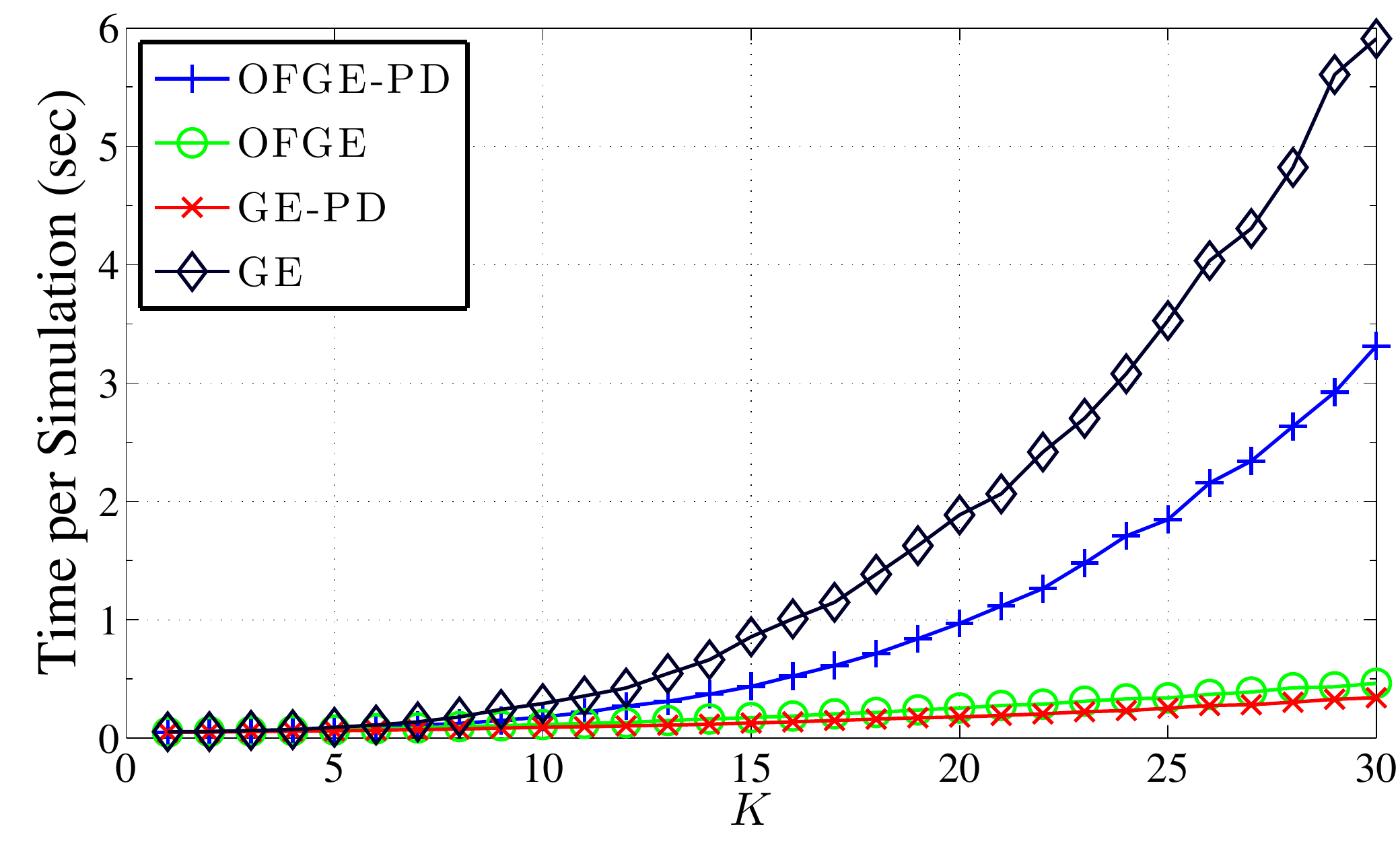}
\caption{Computational cost of the decoding schemes for different numbers of source packets $(K=1,\ldots,30)$.}
%\vspace{-2mm}
\label{fig.time_comp}
\end{figure}

\begin{figure}[t]
\centering
\includegraphics[width=0.95\columnwidth]{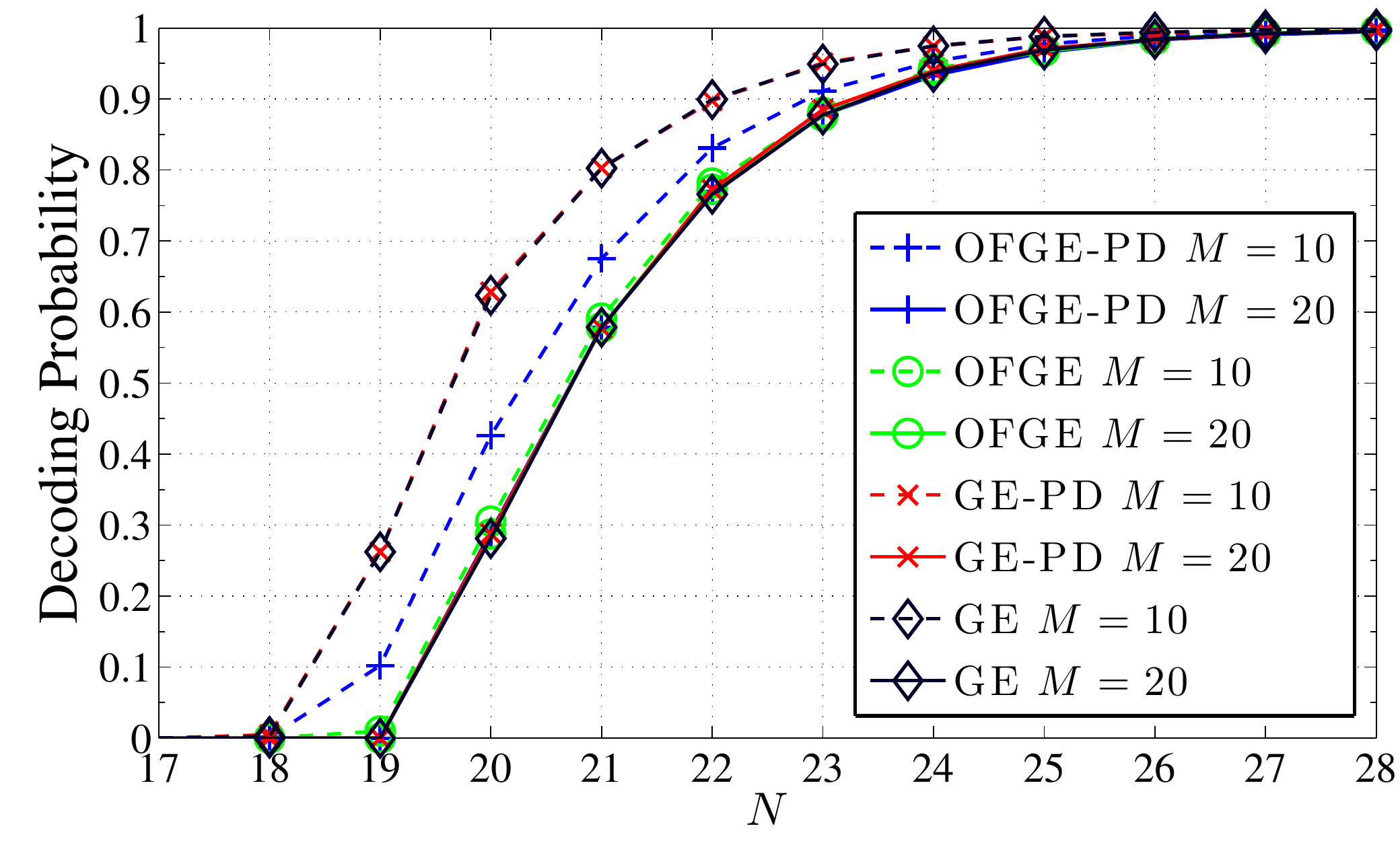}
\caption{Performance comparison of the decoding schemes for $K=20$.}
%\vspace{-2mm}
\label{fig.dec_prob}
\end{figure}

\subsection{Performance Validation of Systematic NC}

In order to validate the derived analytical expressions for the decoding probability of systematic NC, a comparison between theoretical and simulation results was carried out. We considered a source message comprising \mbox{$K=40$} packets, which are encoded using a systematic NC and transmitted over a packet erasure channel with $p=\{0.1, 0.15, 0.3\}$.

Fig.~\ref{fig.sysvalidM} shows that expression \eqref{eq.P_partial} for $\mathrm{P}_{K,M}(N)$ accurately predicts the probability of decoding at least half of the source message \mbox{($M=20$)}. Similarly, expression \eqref{eq.P_full} for $\mathrm{P}_{K}(N)$ matches the simulated results for decoding the entire source message \mbox{($M=40$)}, as reported in Fig.~\ref{fig.sysvalidK}. The excellent agreement between theory and simulation establishes the validity of the theoretical analysis. It also demonstrates that the proposed GE-PD is both efficient and accurate, considering that the number of decoded source packets matches the one predicted by the theoretical model.

\begin{figure}[t]
\vspace{-3.5mm}\centering
\subfloat[$M=20$]{\label{fig.sysvalidM}
\hspace{-1mm}\includegraphics[width=0.5\columnwidth]{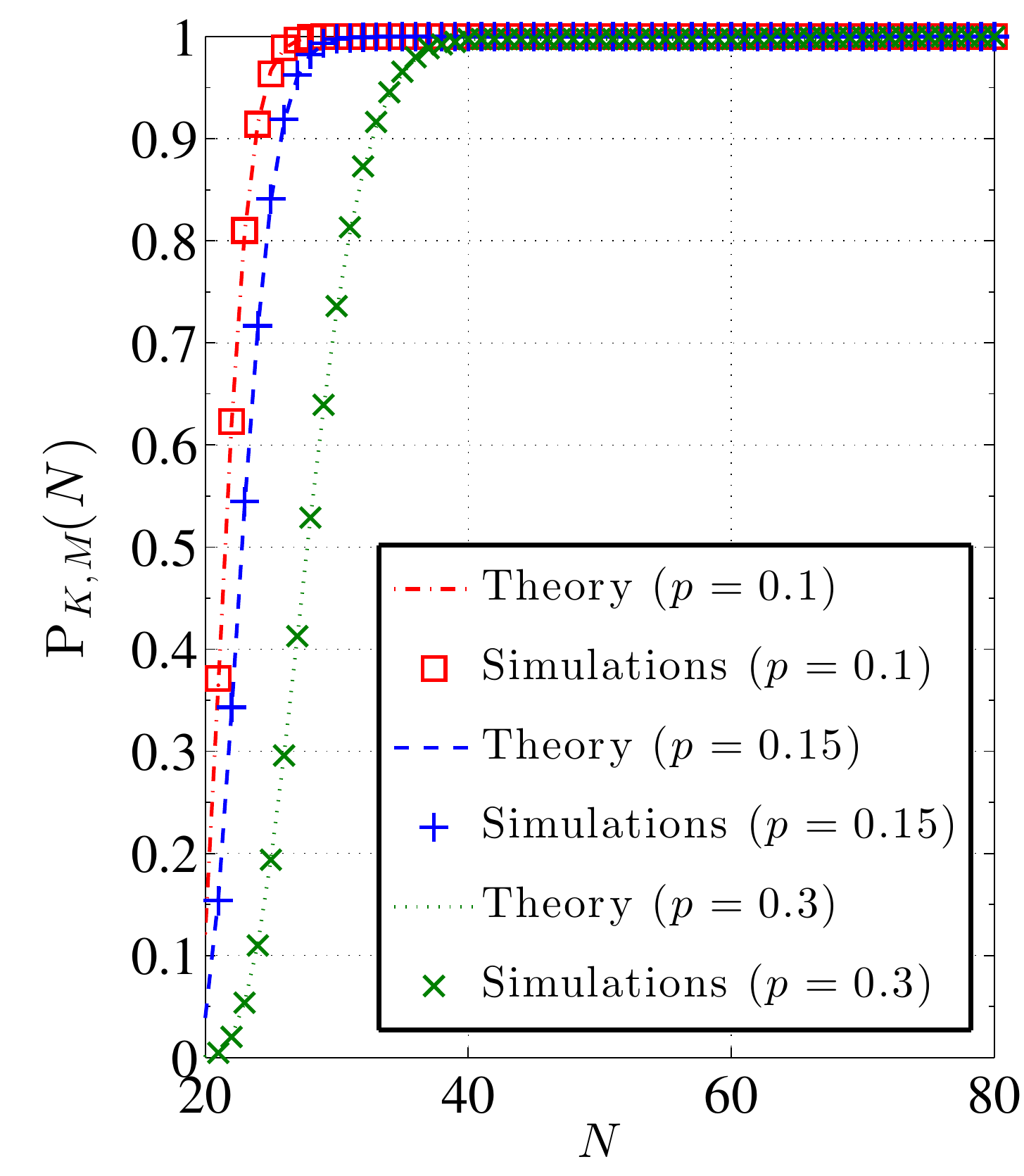}
}
\subfloat[$M=K=40$]{\label{fig.sysvalidK}
\hspace{-2mm}\includegraphics[width=0.5\columnwidth]{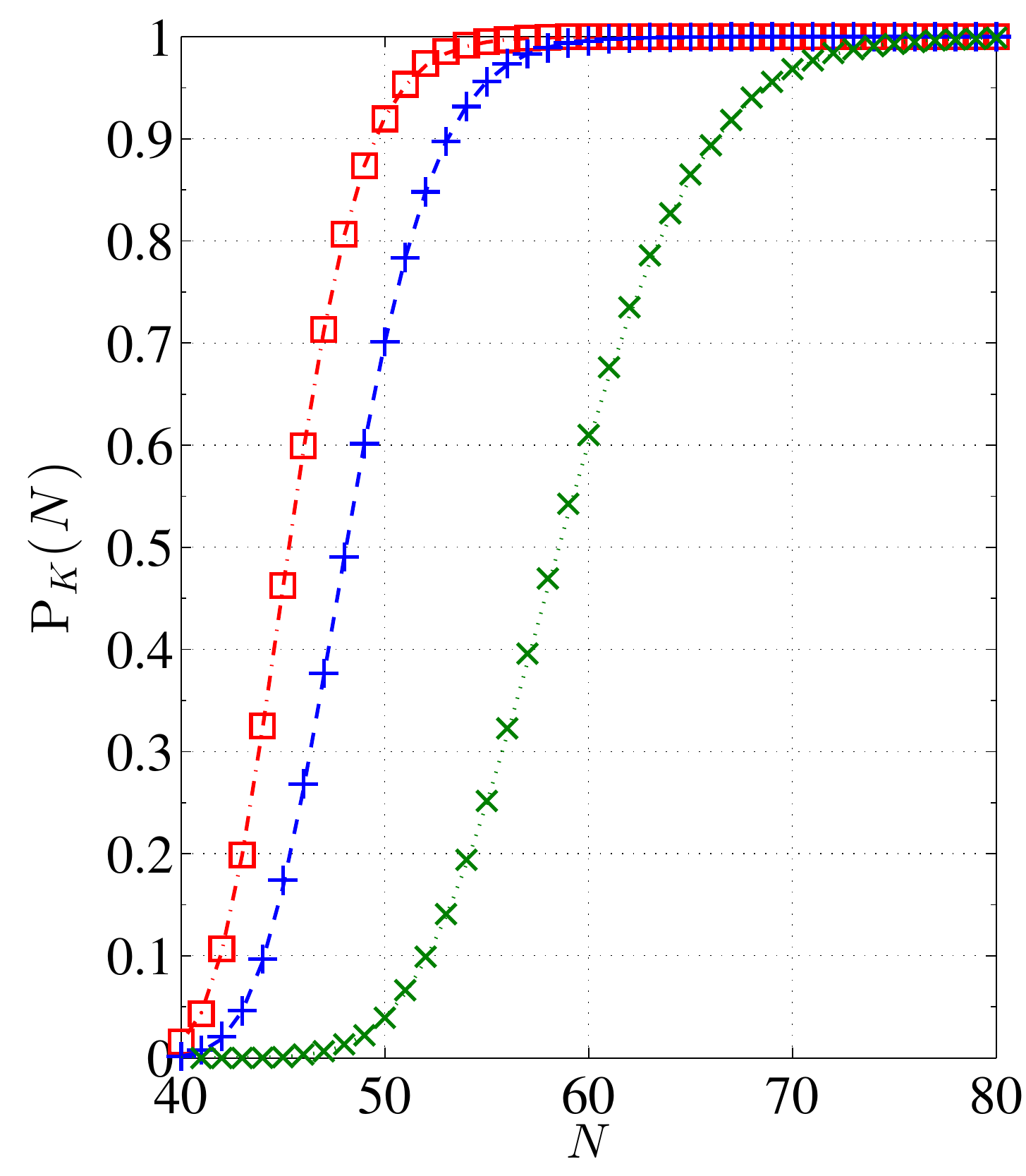}
}
\caption{Performance validation of systematic NC for $K=40$, different values of $p$ and (a) partial recovery ($M=20$) or (b) full recovery ($M=40$) of the source packets.}
\vspace{-2mm}
\label{fig.sysvalid}
\end{figure}

\subsection{Evaluation of Systematic NC for Progressive Decoding}

Fig.~\ref{fig.perfcomp} shows the probability that a receiver employing systematic NC recovers at least half \mbox{$(M=K/2)$} or all \mbox{$(M=K)$} of the source packets, when $N$ packets have been transmitted. The performance of systematic NC is contrasted with that of OU transmission and straightforward NC, referred to here as SF NC for brevity. Two scenarios have been considered; Fig.~\ref{fig.perfcomp_K20} depicts the performance of the three transmission schemes when $K=20$, while Fig.~\ref{fig.perfcomp_K40} presents plots for the case of $K=40$. In both scenarios, the packet erasure probability has been set to $p=0.1$.

We observe in Fig.~\ref{fig.perfcomp_K20} that OU transmission allows the recovery of at least half of the source message for a small value of $N$. However recovery of the whole source message requires a large number of transmitted packets. For example, for a target probability of \mbox{$\hat{P}=0.7$}, a system using OU transmission can retrieve \mbox{$M=10$} source packets if just \mbox{$\hat{N}=11$} packets are transmitted. On the other hand, recovery of all \mbox{$M=20$} source packets requires the transmission of at least 39 packets. In other words, \mbox{$\Delta N=39-11=28$} packets need to be transmitted, on average, to allow recovery of the whole source message, when half of the message has already been retrieved. As we see in Fig.~\ref{fig.perfcomp_K40}, a larger value of $K$ will markedly increase the value of $\Delta N$.

By contrast, SF NC incurs a significant delay in recovering at least half of the source message but only a few extra transmitted packets are required to obtain the entire message. We observe in Fig.~\ref{fig.perfcomp_K20} that if \mbox{$\hat{P}=0.7$} then \mbox{$\hat{N}=24$} packets are needed to reconstruct half of the message, while the transmission of only \mbox{$\Delta N=1$} additional packet is sufficient for the decoding of the entire message.

As is apparent from Fig.~\ref{fig.perfcomp_K20} and Fig.~\ref{fig.perfcomp_K40}, systematic NC combines the best performance characteristics of both OU transmission and SF NC. We observe that the value of $\hat{N}$ for recovering at least half of the source packets is as small as that of OU transmission, while the required number of transmitted packets for retrieving all of the source packets is smaller than or similar to that of SF FC. The latter observation confirms Proposition \ref{prop.sysNC}. Consequently, systematic NC is the most appropriate of the considered transmission schemes for progressive packet decoding, as it exhibits a high probability of either partially or fully decoding the source message.

\begin{figure}[t]
%\vspace{-3.5mm}
\centering
\subfloat[$K=20$]{\label{fig.perfcomp_K20}
\hspace{-1mm}\includegraphics[width=0.88\columnwidth]{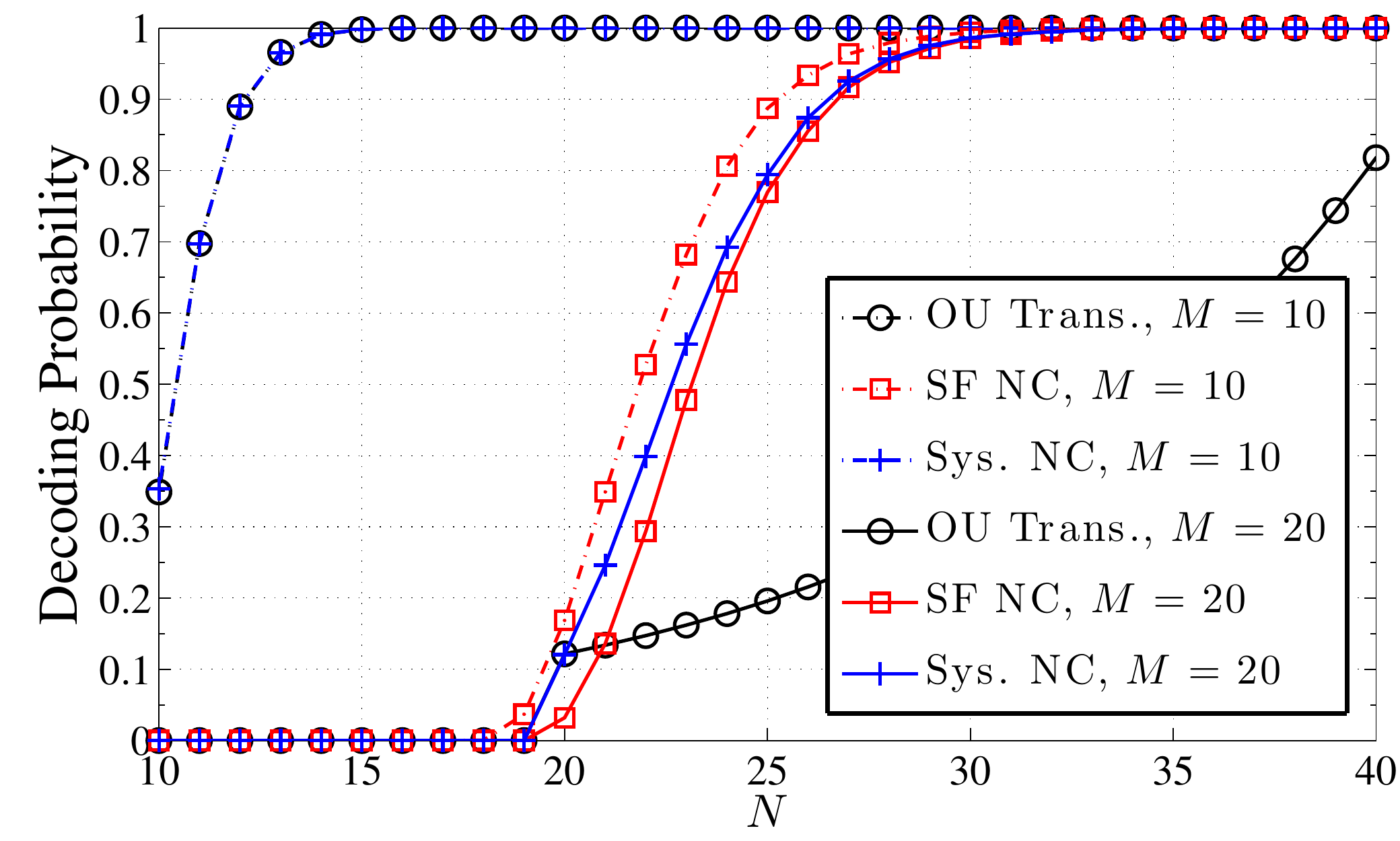}
}\\
\vspace{-2mm}
\subfloat[$K=40$]{\label{fig.perfcomp_K40}
\hspace{-2mm}\includegraphics[width=0.88\columnwidth]{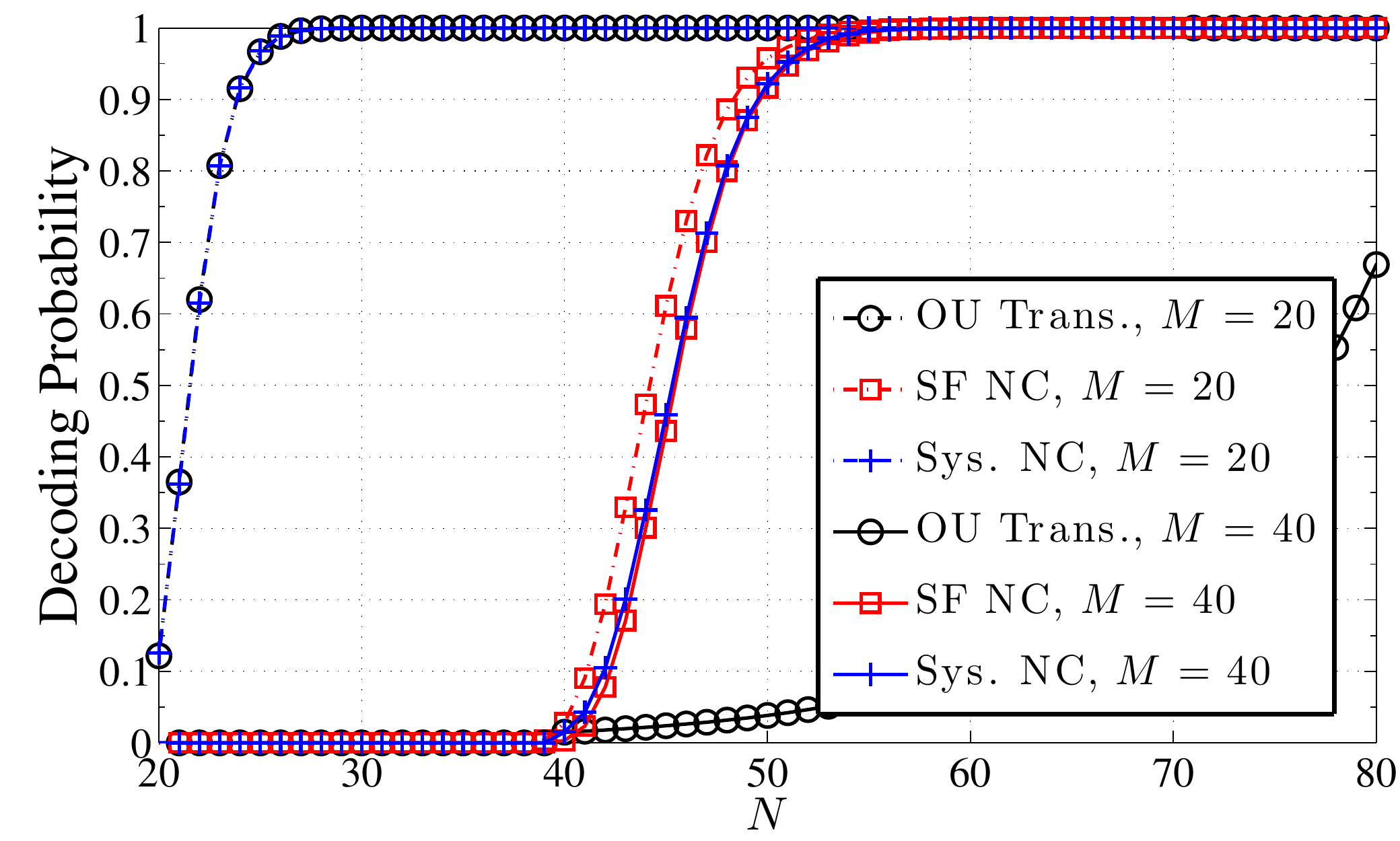}
}
\caption{Decoding probabilities as a function of $N$ for $p=0.1$ and (a) $K=20$ or (b) $K=40$.}
\vspace{-2mm}
\label{fig.perfcomp}
\end{figure}

% -------------------------------------

\section{Conclusions}
\label{sec:conclusions}

In this paper, we considered systematic random linear network coding, obtained theoretical expressions that accurately describe its decoding probability and proved that systematic network codes exhibit a higher probability of decoding the entirety of a source message than straightforward network coding. We also proposed Gaussian elimination for Progressive Decoding (GE-FD), which aims to recover source packets as soon as one or more transmitted packets are successfully delivered to a receiver. We demonstrated that GE-PD performs similarly to the optimal theoretical decoder in terms of decoding probability and also exhibits low computational cost. Furthermore, we established that the decoding delay characteristics of systematic network coding for both partial and full recovery of source messages are notably better than those of straightforward network coding.

% -------------------------------------

\section*{Acknowledgment}
%This work was conducted as part of the R2D2 project, which is supported by the Engineering and Physical Sciences Research Council (EPSRC) under Grant EP/L006251/1.

This work was conducted as part of the R2D2 project, which is supported by EPSRC under Grant EP/L006251/1.

\bibliographystyle{IEEEtran}
\bibliography{ICC15bib}
\vspace{-4mm}
\end{document}